\documentclass[letterpaper, 10 pt, journal, twoside]{ieeetran} 
\usepackage{etex}
\usepackage{graphicx}
\usepackage{epstopdf}
\usepackage{booktabs}
\usepackage{amsfonts}
\usepackage{amsthm}
\usepackage{textcomp}
\usepackage{amsmath} 
\usepackage{amssymb}
\usepackage[colorlinks=true, urlcolor=blue, citecolor=blue, linkcolor=blue]{hyperref}
\usepackage{balance}
\usepackage{pbox}
\usepackage{psfrag}
\usepackage{algpseudocode}
\usepackage{algorithm}
\usepackage{tikz}
\usepackage{tikz-qtree}
\usetikzlibrary{trees}
\usetikzlibrary{fit,calc,arrows,positioning}
\usepackage{pgfplots}
\pgfplotsset{compat=newest}
\pgfplotsset{plot coordinates/math parser=false}
\newlength\figureheight
\newlength\figurewidth
\usepackage{graphicx}
\usepackage{subcaption}
\usepackage{caption}
\usepackage{mathtools}
\usetikzlibrary{shapes,arrows}
\makeatletter
\def\BState{\State\hskip-\ALG@thistlm}
\makeatother
\usepackage{tikz}
\usetikzlibrary{fit,calc,arrows,positioning}
\pgfplotsset{compat=newest}
\pgfplotsset{plot coordinates/math parser=false}
\newcommand{\range}{\operatorname{range}}

\newcommand{\thetatilde}{\tilde{\theta}}

\newcommand{\R}{\mathbb{R}}
\newcommand{\C}{\mathcal{C}}
\newcommand{\nats}{\mathbb{N}}

\newcommand{\dom}{\operatorname{dom}}

\newcommand{\Id}{{I}}

\newcommand*{\tr}{^{\mkern-1.5mu\mathsf{T}}}

\newcommand{\minimize}{\operatorname{minimize}}
\newcommand{\Tr}{\operatorname{trace}}

\newcommand{\He}{\operatorname{He}}

\newcommand{\Spn}{\mathbb{S}^n_+}

\newcommand{\Sspnz}{\mathbb{S}^{n_z}_{+}}
\newcommand{\Spnz}{\mathbb{S}^{n_z}_{+}}
\newcommand{\Spny}{\mathbb{S}^{n_y}_{+}}
\newcommand{\0}{{0}}

\newcommand{\ep}{{\varepsilon}}

\newtheorem{theorem}{Theorem}

\newtheorem{definition}{Definition}

\newtheorem{assumption}{Assumption}

\newtheorem{proposition}{Proposition}
\newtheorem{remark}{Remark}

\newcommand{\source}{{THIS IS A PREPRINT VERSION. IF YOU FOUND THIS READING USEFUL FOR YOUR RESEARCH PLEASE CITE THE PUBLISHED VERSION DOI: \href{https://doi.org/10.1109/LCSYS.2021.3081345}{https://doi.org/10.1109/LCSYS.2021.3081345}}}
 \usepackage{fancyhdr}
\makeatletter

\def\ps@IEEEtitlepagestyle{}
\title{\LARGE \bf Observer Design for Linear Aperiodic Sampled-Data Systems: A Hybrid Systems Approach}
\author{Francesco~Ferrante, \IEEEmembership{Member, IEEE} and Alexandre Seuret \thanks{Francesco Ferrante is with Univ. Grenoble Alpes, CNRS,  GIPSA-lab, F-38000 Grenoble, France. Email: francesco.ferrante@gipsa-lab.fr} 
\thanks{Alexandre Seuret is with LAAS-CNRS, Universit\'e de Toulouse, CNRS, Toulouse, France. Email: aseuret@laas.fr}
\thanks{Research by Francesco Ferrante is partially funded by ANR via project HANDY, number ANR-18-CE40-0010.}
}
 
\begin{document}
\maketitle\begin{abstract}
Observer design for linear systems with aperiodic sampled-data measurements is addressed. To solve this problem, a novel hybrid observer is designed. The main peculiarity of the proposed observer consists of the use of two output injection terms, one acting at the sampling instants and one providing an intersample injection. The error dynamics are augmented with a timer variable triggering the arrival of a new measurement and analyzed via hybrid system tools. Using Lyapunov theory, sufficient conditions for the convergence of the observer are provided. Relying on those conditions, an optimal LMI-based design is proposed for the observer gains. The effectiveness of the approach is illustrated in an example.
\end{abstract}
\begin{IEEEkeywords}
Hybrid systems, sampled-data systems, LMIs, Observer Design.
\end{IEEEkeywords}
\section{Introduction}
\subsection{Motivation}
\IEEEPARstart{S}{tate} estimation is a fundamental problem in systems and control theory. Indeed, since state variables can be difficult or impossible to measure, having access to reliable estimates of the plant state is paramount for fault detection, monitoring, and control.
The pervasive use of data networks in modern control systems has led to several major difficulties in the design of reliable observers for networked systems. Indeed, when the plant output is accessed through a data network, the typical assumption of continuously or periodically measuring is unrealistic; see, e.g., \cite{hespanha2007survey,hristu2005handbook,Walsh2002}, {\cite{Feketa:2013:impulsive} and
\cite{hetel2017recent} for a recent survey on aperiodic sampled-data systems. In this paper, we are interested in the design of state observers in the presence of sporadically available measurements. The fact that measurements are available only at some aperiodic isolated times requires the use of observer schemes that are able to handle this intermittent stream of information to generate suitable innovation terms. This naturally leads to the use of hybrid observers, whose dynamics evolve continuously when no measurements are available and experience instantaneous changes when a new measurement gets available.   
\subsection{Problem Formulation}
\label{sec:ProblemStatement}
In this paper, we consider the problem of estimating the state of a continuous-time linear time-invariant plant in the presence of intermittent measurements. In particular, we consider a plant of the form:
\begin{equation}
\label{eq:P2:Chap3:Plant}
\left\{
\begin{array}{lcl}
\dot{z}&=&Az\\
y&=&Cz
\end{array}
\right.
\end{equation}
where $z\in\mathbb{R}^{n_z}$ is the plant state and $y\in\mathbb{R}^{n_y}$ is the plant output, with $n_z,n_y\in\mathbb{N}_{>0}$. Matrices $A$ and $C$ are known, constant and of appropriate dimensions. The plant output $y$ is assumed to be available only at some time instants $t_k$, $k\in\mathbb{N}_{>0}$, not known {\itshape a priori}. We assume that the sequence $\{t_k\}_{k\in\mathbb{N}_{>0}}\!$ is unbounded, in addition we suppose that there exist two positive real scalars $T_1\leq T_2$  such that  
\begin{equation}
\begin{array}{lr}
\label{eq:P2:Chap3:timebound}
0\leq t_{1}\leq T_2,\quad T_1\leq t_{k+1}-t_{k}\leq T_2,\quad\forall k\in\mathbb{N}.
\end{array}
\end{equation}
The lower bound in condition \eqref{eq:P2:Chap3:timebound} prevents the existence of accumulation points in the sequence $\{t_k\}_{k\in\mathbb{N}_{>0}}$, and, hence, avoids the existence of Zeno behaviors, which are typically undesired in practice.
In fact, $T_1$ defines a strictly positive minimum time in between consecutive measurements. Furthermore, $T_2$ defines the \emph{Maximum Allowable Transfer Time (MATI)} \cite{postoyan2012framework}.
\subsection{Related work}
The design of the observers in the presence of sporadic measurements has been largely studied by researchers over the last two decades and several observer design strategies have been proposed in the literature. Such strategies essentially belong to two main families. The first one pertains to observers whose state is entirely reset whenever a new measurement is available and that run in open-loop in between such events  \cite{Andrieu:2013aa,etienne2017observer,Ferrante2016state}, \cite{Mazenc:2014:construction}, \cite{nadri2003design,sferlazza2018time}, i.e., continuous-discrete observers:
\begin{equation}
\label{eq:CT_OBS}
\left\{\!\!\!\!\!\begin{array}{ll}
\begin{array}{rcl}
\dot{\hat{z}}(t)&\! =& \!\!\!A\hat{z}(t)\\
\end{array} &\hspace{-0.1cm}\mbox{if}\,\, t\neq t_k, k\in\nats_{>0},\\
\begin{array}{rcl}
\hat{z}(t^{+}) &\!\!\!\! \!= &\!\!\!\hat{z}(t)\!+\!F(y(t)\!-\!C\hat{z}(t))\\
\end{array}&
\hspace{-0.1cm}\mbox{if}\,\, t=t_k, k\in\nats_{>0},
\end{array}\right.
\end{equation}
where $F$ is a gain to be designed, which can be potentially selected to be dependent on the time elapsed in between measurements; see \cite{Andrieu:2013aa,etienne2017observer,sferlazza2018time}. The working principle of the above observer is as follows, when no plant measurement is available, the observer behaves as a copy of the plant. When a new measurement gets available, the observer state is instantaneously reset. The main advantage of this class of observers is that it allows to achieve fast convergence rate. On the other hand, fast convergence rate typically comes at the price of enforcing large changes of the observer state at the sampling times. This can be unsuitable when employing an observer-based control. Indeed, large jumps in the estimate may lead to overly large discontinuities in the control input, which can jeopardize the safety of the actuator. In addition, the fact that the observer runs in ``open-loop'' may lead to poor inter-sample behaviors. 

A completely different paradigm has been proposed by  Karafyllis and Kravaris in \cite{karafyllis2009continuous}. In \cite{karafyllis2009continuous}, the  proposed architecture is composed by 
a so-called \emph{output predictor} whose state is reset to the value of the plant output at the sampling times and used as an inter-sample injection to feed a Luenberger-like observer. A similar approach for control design is presented in
 \cite{Ahmed-AliExponential2016}.
This idea has been later generalized in \cite{ferrante2018TAC}. The main advantage of this class of observers is that they avoid the occurrence of jumps in the estimate. Moreover, the above mentioned intersample injection can be tuned to conveniently shape the transient response. However, this class of observers typically exhibit less aggressive transient performance when compared to the scheme in \eqref{eq:CT_OBS}. \subsection{Outline of the Proposed Solution}
With the objective of achieving a tradeoff between convergence speed and transient performance, while avoiding overly large jumps in the plant estimate, in this paper we blend the architecture \eqref{eq:CT_OBS} with that in \cite{ferrante2018TAC} and propose a  
new class of hybrid observers for aperiodic sampled-data systems. In particular, we consider the following hybrid observer:
\begin{equation}
\label{eq:P2:Chap3:ObsSampleHold}
\begin{array}{ll}
\left\{
\begin{array}{rcl}
\dot{\hat{z}}(t)& =& A\hat{z}(t)+L\theta(t)\\
\dot{\theta}(t)& =&H\theta(t)
\end{array}\right. &\hspace{-0.2cm}\text{if}\,\, t\neq t_k, k\in\mathbb{N}_{>0},\\
\left\{
\begin{array}{rcl}
\hat{z}(t^{+}) &\!\!\!\! = &\hat{z}(t)+F(y(t)-C\hat{z}(t))\\
\theta(t^{+})&\!\!\!\!= &(I-CF)(y(t)-C\hat{z}(t))
\end{array}\!\!\!\! \right.&\hspace{-0.2cm}
\text{if}\,\, t=t_k,k\in\mathbb{N}_{>0},
\end{array}
\end{equation}
where the observer gains $L$, $F$, and $H$ are real matrices of appropriate dimensions to be designed. Variable $\hat{z}$ represents the estimate of $z$ provided by the observer.
The observer in \eqref{eq:P2:Chap3:ObsSampleHold} generalizes several existing architectures for state estimation in the presence of sampled-data aperiodic measurements. In particular, selecting $H=0$ and $F=0$ leads to classical sampled-data observers with zero-order hold output injection \cite{montestruque2002state,seuret2006networked}. If only $F$ is set to zero, the resulting observer reduces to the observer presented in \cite{ferrante2018TAC}. If $H$ and $L$ are both set equal to zero, one recovers \eqref{eq:CT_OBS}.
\subsection{Contribution and organization}
The main contribution of this paper consists of sufficient conditions for the design of the observer \eqref{eq:P2:Chap3:ObsSampleHold} to ensure global exponential stability of the estimation error with tunable transient performance. Compared to the previous schemes in this area, the observer contains three correction terms to be designed. More precisely, a blend of injections during flows and jumps. 
The paper is organized as follows. Section~\ref{sec:Construction} presents a hybrid model of the error dynamics and a sufficient conditions to exponential stability of the estimation error. The main contributions of the paper are presented in Section~\ref{sec:ObDes}, where computationally affordable conditions for the design of the observer gains are provided. These results are illustrated through an example in Section~\ref{sec:example}.
\subsection{Notation}
The symbol $\nats$ stands ($\nats_{>0}$) for the set of nonnegative (positive) integers, $\R_{\geq 0}$ ($\R_{>0}$) denotes the set of nonnegative (positive) reals, $\R^n$ is the $n$-dimensional Euclidean space, $\R^{n\times m}$ is the set of $n\times m$ real matrices, and $\Spn$ is the set of $n\times n$ symmetric positive definite matrices. The identity matrix is denoted by $\Id$. The symbol $M\tr$ denotes the transpose of the matrix $M$. When $M$ is a square matrix, $\He (M)=M+M\tr$. For a symmetric matrix $M$, $M\succ(\prec)\,0$ and $M\succeq(\preceq)\,0$ indicate that $M$ is positive (negative) definite and positive (negative) semidefinite, respectively. The symbols $\lambda_{\min}(M)$ and  $\lambda_{\max}(M)$ denote, respectively, the largest and the smallest eigenvalue of $M$. In partitioned symmetric matrices, the symbol $\bullet$ stands for symmetric blocks. For $x\in\mathbb{R}^n$, $\vert x \vert$ denotes its Euclidean norm. 
The equivalent notation $(x, y)=[x\tr\,\,y\tr]\tr$ is used for vectors.
Given $x\in\mathbb{R}^{n}$  and $\mathcal{A}\subset \mathbb{R}^{n}$ nonempty, the distance of $x$ to $\mathcal{A}$ is defined as 
$\vert x \vert_{\mathcal{A}}=\inf_{y\in {\mathcal{A}}} \vert x-y \vert$. For any function $z\colon\mathbb{R}\rightarrow\mathbb{R}^n$, we denote $z(t^+)\coloneqq \lim_{s\rightarrow t^+} z(s)$, when it exists.
\subsection{Preliminaries on Hybrid Dynamical Systems}
\label{sec:preliminariesHybridInc}
In this paper we consider hybrid  dynamical systems in the framework \cite{goebel2012hybrid} represented as:
\begin{equation}
\label{eq:HybridPlant}
\mathcal{H}\left\{
\begin{array}{lcll}
\dot{x}&=&f(x),&\quad x\in \mathcal C,\\
x^+&\in&G(x),&\quad x\in  \mathcal D.
\end{array}\right.
\end{equation}
where $x\in\R^{n}$ is the state vector, $f\colon\R^{n}\rightarrow\R^{n}$ denote the \emph{flow map} and $G\colon\R^{n}\rightrightarrows\R^{n}$ the (set valued) \emph{jump map}, while the sets $\mathcal C\subset\R^{n}$ and $\mathcal D\subset\R^{n}$ refer to the \emph{flow} and the \emph{jump sets}, respectively. A set $E\subset\R_{\geq 0}\times \nats$ is a \emph{hybrid time domain} (\emph{HTD}) if it is the union of a finite or infinite sequence of intervals $[t_j, t_{j+1}]\times\{j\}$, with the last interval (if existent) of the form $[t_j,T)$ with $T$ finite or $T=\infty$. A function $\phi\colon\dom\phi\rightarrow\R^n$ is a hybrid arc if $\dom\phi$ is a HTD and $t\mapsto\phi(t, j)$ is locally absolutely continuous for each $j$. Given a hybrid arc $\phi$, 
$\dom_t \phi\coloneqq \{t\in\R_{\geq 0}\colon \exists j\in\nats\,\,\mbox{s.t.}\,\,(t, j)\in\dom \phi\}$ and $\dom_j \phi\coloneqq \{j\in\nats\colon \exists t\in\mathbb{R}_{\geq 0}\,\,\mbox{s.t.}\,\,(t, j)\in\dom \phi\}$.
Given a hybrid arc $\phi$, $s\in\dom_t \phi$, and $i\in\dom_j \phi$, $j(s)\coloneqq\min\{j\in\nats\colon\,\, (s, j)\in\dom \phi\}$ and $t(i)\coloneqq\min\{t\in\R_{\geq 0}\colon\,\, (t, i)\in\dom \phi\}$. A solution to \eqref{eq:HybridPlant} is any hybrid arc that satisfies its dynamics. A solution $\phi$ to $\mathcal{H}$ is maximal if its domain cannot be extended and it is complete if its domain is unbounded. Given a set $\mathcal M$, we denote by $\mathcal{S}_{\mathcal{H}}(\mathcal M)$ the set of all maximal solutions $\phi$ to $\mathcal{H}$ with $\phi(0,0)\in \mathcal M$. If no set $\mathcal{M}$ is mentioned, $\mathcal{S}_{\mathcal{H}}$ is the set of all maximal solutions to $\mathcal{H}$.

The following notion of global exponential stability is considered in the paper.
\begin{definition}(Global exponential stability~\textsc{\cite{teel2013lyapunov}})
\label{GES}
Let $\mathcal{A}\subset\mathbb{R}^{n}$ be closed. The set $\mathcal{A}$ is said to be \emph{globally exponentially stable} (GES) for hybrid system $\mathcal{H}$ if there exist strictly positive real numbers $\lambda, k$ such that every maximal solution $\phi$ to $\mathcal{H}$ is complete and it satisfies for all $(t, j)\in\dom\phi$
\begin{equation}
\label{eq:GESBound}
\vert\phi(t,j)\vert_{\mathcal{A}}\leq ke^{-\lambda(t+j)}\vert\phi(0,0)\vert_{\mathcal{A}}.
\end{equation}
\end{definition}

We invite the reader to check \cite{goebel2012hybrid} for more details on the considered framework for hybrid systems.
\section{Hybrid Modeling and Stability Analysis}
\label{sec:Construction}
\subsection{Hybrid Modeling}
Let us first introduce the following change of variables $$\varepsilon\coloneqq z-\hat{z},\quad 
\tilde{\theta}\coloneqq C(z-\hat{z})-\theta,
$$
which defines, respectively, the estimation error and the difference between the output estimation error and $\theta$. In particular, the dynamics of those estimation errors read:
\begin{equation}
\scalebox{1}{$
\begin{array}{ll}
\left\lbrace
\begin{array}{ll}
\begin{bmatrix}
\dot{\ep}(t)\\
\dot{\tilde{\theta}}(t)
\end{bmatrix}&=\mathsf{F}\begin{bmatrix}
\ep(t)\\
\tilde{\theta}(t)
\end{bmatrix}\end{array}\right. &\text{if}\,\,t\neq t_k, k\in\mathbb{N}\\
\left\lbrace
\begin{array}{ll}
\begin{bmatrix}
\ep(t^+)\\
\tilde{\theta}(t^+)
\end{bmatrix}&=\mathsf{G}\begin{bmatrix}
\ep(t)\\
\tilde{\theta}(t)
\end{bmatrix}
\end{array}\right.&
\text{if}\,\, t=t_k, k\in\mathbb{N}
\end{array}
\label{eq:P2:Chap3:ObsGenErr}$}
\end{equation}
where 
\begin{equation}
\label{eq:P2:Chap3:FG}
\begin{array}{ll}
\mathsf{F}\!\coloneqq\!\!\begin{bmatrix}
A\!-\!LC&\!\!L\\
CA\!-\!CLC\!-\!HC&\!\!CL\!+\!H
\end{bmatrix}\!\!,\ 
\mathsf{G}\!\coloneqq\!\!\begin{bmatrix}
\Id\!-\!FC&\!\!\! 0\\
0&\!\!\! 0
\end{bmatrix}\!\!.
 \end{array}
\end{equation}
The fact that the observer experiences jumps, when a new measurement is available and evolves according to a differential equation in between updates, suggests that the updating process of the error dynamics can be described via a hybrid system.  Hence, we represent the whole system composed by the plant \eqref{eq:P2:Chap3:Plant}, the observer \eqref{eq:P2:Chap3:ObsSampleHold}, and the logic triggering jumps as a hybrid system. The proposed hybrid systems approach also models the hidden time-driven mechanism triggering the jumps of the observer. To this end, and as in \cite{Ferrante2016state,merco2019resiliency,ferrante2018TAC}, we augment the state of the system  with an auxiliary timer variable $\tau$ that keeps track of the duration of flows and triggers a jump whenever a certain condition is verified. This additional state allows to describe the time-driven triggering mechanism as a state-driven triggering mechanism, thereby leading to a model that can be efficiently represented by relying on the framework for hybrid systems in \cite{goebel2012hybrid}.
More precisely, we make $\tau$ decrease as ordinary time $t$ increases and, whenever $\tau=0$, reset it to any point in $[T_1,T_2]$, so as to enforce \eqref{eq:P2:Chap3:timebound}. After each jump, we allow the system to flow again.  
The whole system composed by the states $\ep$ and $\tilde{\theta}$, and the timer variable $\tau$ can be represented by the following hybrid system, which we denote by $\mathcal{H}_e$, with state
$x=(\ep,\tilde{\theta},\tau)\in\R^{n_x}$
where  $n_x\coloneqq n_z+n_y+1$:
\begin{subequations}
\label{eq:P2:Chap3:ObsHybrid}
\begin{equation}
\begin{array}{l}
\mathcal{H}_e\left\{\
\begin{array}{rclrcl}
\dot{x}&=&f(x),&\quad &x\in \mathcal{C},\\
x^+&\in&G(x),&&x\in \mathcal{D},
\end{array}
\right.\\
\end{array}
\end{equation}
where
\begin{equation}
\label{eq:P2:Chap3:FlowMap}
f(x)\coloneqq\begin{bmatrix}
\mathsf{F}\begin{bmatrix}
\ep\\
\tilde{\theta}
\end{bmatrix}\\
-1
\end{bmatrix},\qquad\forall x\in \mathcal{C},\end{equation}
\begin{equation}
\label{eq:P2:Chap3:JumpMap}
G(x)\coloneqq\begin{bmatrix}
	\mathsf{G}\begin{bmatrix}
		\ep\\
		\tilde{\theta}
	\end{bmatrix}\\
	[T_1,T_2]
\end{bmatrix},\qquad\forall x\in \mathcal{D},
\end{equation}
and the flow set $\mathcal{C}$ and the jump set $\mathcal{D}$ are defined as follows
\begin{equation}
\label{eq:P2:Chap3:Sets2}
\mathcal{C}\coloneqq\mathbb{R}^{n_z+n_y}\times[0,T_2],\quad\mathcal{D}\coloneqq\mathbb{R}^{n_z+n_y}\times\{0\}.
\end{equation}
\end{subequations}
The set-valued jump map allows to capture all possible sampling events fulfilling \eqref{eq:P2:Chap3:timebound}.
Specifically, the hybrid model in \eqref{eq:P2:Chap3:ObsHybrid} is able to characterize not only the behavior of the analyzed system for a given sequence $\{t_k\}_{k=1}^\infty$, but for any sequence satisfying \eqref{eq:P2:Chap3:timebound}. 
Concerning existence of solutions to system \eqref{eq:P2:Chap3:ObsHybrid}, by relying on the concept of solution proposed in \cite[Definition 2.6]{goebel2012hybrid}, it is straightforward to check that any maximal solution to \eqref{eq:P2:Chap3:ObsHybrid} is complete. Thus, completeness of the maximal solutions to \eqref{eq:P2:Chap3:ObsHybrid} is guaranteed for any choice of the gains $L, H$, and $F$. In addition, we can characterize the domain of these solutions. 
In particular, from the definition of the sets $\mathcal{C}$ and $\mathcal{D}$, it follows that for any  maximal solution $\phi$ to $\mathcal{H}_e$, 
$\dom \phi=\displaystyle\bigcup_{j\in\nats}([t_j,t_{j+1}])\times\{j\},$
with $t_0=0$, $0\leq t_{1}\leq T_2$, and $t_{j+1}-t_j\in[T_1 ,\ T_2]$, for all $j\in\nats_{>0}$.

To solve the considered state estimation problem, our approach is to design gains $L, F$, and $H$ in \eqref{eq:P2:Chap3:ObsHybrid} such that the set wherein the estimation error is zero is globally exponentially stable for \eqref{eq:P2:Chap3:ObsHybrid}. To this end, we consider the following closed set
\begin{equation}
\label{eq:P2:Chap3:A}
\mathcal{A}=\{0\}\times \{0\}\times[0,\ T_2],
\end{equation}
and provide sufficient conditions to ensure that $\mathcal{A}$ is GES for system $\mathcal{H}_e$.
\subsection{Sufficient conditions for exponential stability}
In this section, sufficient conditions for observer design are provided. To this end, let us consider the following assumption whose role will be clarified later via Theorem~\ref{theorem:P2:Chap3:Main}.
\begin{assumption}
\label{Assumption:L2}
There exist two continuously differentiable functions $V_1\colon\mathbb{R}^{n_z+1}\rightarrow \R$, $V_2\colon\R^{n_y+1}\rightarrow \R$, positive real numbers $\alpha_1,\alpha_2,\omega_1,\omega_2$, $\chi_c$, and $\varpi_d\in[0, 1)$ such that
\bigskip
\begin{itemize}
\item[(A1)]$\alpha_1\vert \ep\vert^2\leq V_1(\ep, \tau)\leq \alpha_2\vert \ep\vert^2\qquad \forall x\in \mathcal{C}$;
\bigskip

\item[(A2)]$\omega_1\vert \thetatilde\vert^2\leq V_2(\thetatilde, \tau)\leq \omega_2\vert \thetatilde\vert^2\qquad \forall x\in \mathcal{C}$;
\medskip

\item[(A3)] for each $\ep\in\R^{n_z}, \nu\in[T_1, T_2]$ 
\begin{equation}
\label{eq:A4}
V_1((I-FC)\ep, \nu)\leq (1-\varpi_d) V_1(\ep,0),
\end{equation}

\item[(A4)] for each $x\in \mathcal{C}$, the function $x\mapsto V(x)\coloneqq V_1(\ep, \tau)+V_2(\thetatilde,\tau)$ is such that
\medskip 
\begin{equation}
\label{eq:A3}
\langle\nabla V(x), f(x)\rangle\leq -2\chi_c V(x).
\end{equation}
\hfill $\triangle$
\end{itemize}
\end{assumption}

The following result provides a sufficient condition for global exponential stability of the set $\mathcal{A}$ defined in \eqref{eq:P2:Chap3:A}.
\begin{theorem}\label{thm:GES}
\label{theorem:P2:Chap3:Main}
Let Assumption~\ref{Assumption:L2} hold. Then, the set $\mathcal{A}$ in~\eqref{eq:P2:Chap3:A} is globally exponentially stable (GES) for  $\mathcal{H}_e$.
\end{theorem}
\begin{proof}
Using items (A2) and (A3) in Assumption~\ref{Assumption:L2}, one has that for all $x\in\mathcal{D}, g\in G(x)$
\begin{equation}
V(g)\leq e^{-2\chi_d} V(x),
\label{eq:DeltaV}
\end{equation}
where
$\chi_d\coloneqq -\frac{1}{2}\ln(1-\varpi_d)\geq 0$. Let $\phi$ be any maximal solution to $\mathcal{H}_e$. Then, by integrating $(t, j)\mapsto (V\circ\phi)(t, j)$ and using 
item (A4) in Assumption~\ref{Assumption:L2} and \eqref{eq:DeltaV}, one has, for all $(t,j)\in\dom\phi$, 
$V(\phi(t,j))\leq e^{-2(\chi_c t+\chi_d j) }V(\phi(0,0)),$
which by using items (A1) and (A2) in Assumption~\ref{Assumption:L2} yields:
\begin{equation}
\vert\phi(t,j)\vert_\mathcal{A}\leq e^{-(\chi_c t+\chi_d j) }\frac{\rho_2}{\rho_1}\vert\phi(0,0)\vert_\mathcal{A},
\quad\forall(t,j)\in\dom\phi,
\label{eq:phi_t_ges}
\end{equation}
where $\rho_1\coloneqq\min\{\alpha_1,\omega_1\}$ and 
$\rho_2\coloneqq\max\{\alpha_2,\omega_2\}$. To conclude, using \cite[Lemma 1]
{Ferrante2018TACReport}, it follows that there exist some solution independent positive real numbers $\varrho$ and  $\lambda$ such that for all $(t, j)\in\dom\phi$, $-\chi_c t\leq \varrho-\lambda (t+j)$. Hence, by using the bound in \eqref{eq:phi_t_ges}, one gets, for all $(t,j)\in\dom\phi$,
$\vert\phi(t,j)\vert_\mathcal{A}\leq e^{-\lambda(t+j)}e^{\varrho}\frac{\rho_2}{\rho_1}\vert\phi(0,0)\vert_\mathcal{A}$.
This concludes the proof.
\end{proof}
\begin{remark}
It is worth to mention that due to $\C$ and $\mathcal{D}$ being closed, $f$ being continuous, and $G$ being outer semicontinuous, hybrid system \eqref{eq:P2:Chap3:ObsHybrid} satisfies the so-called \emph{hybrid conditions} and so it is well posed in the sense of  \cite[Definition 6.29]{goebel2012hybrid}.
Well posedness of \eqref{eq:P2:Chap3:ObsHybrid} ensures that the stability property established in Theorem~\ref{theorem:P2:Chap3:Main} enjoys desirable robustness features that are well characterized in \cite[Ch. 7]{goebel2012hybrid}.
\end{remark}
\subsection{Quadratic conditions}
A possible construction for the functions $V_1$ and $V_2$ in Theorem~\ref{thm:GES} is illustrated in the result given next.
\begin{theorem}
\label{Theorem1bis}
Let $L,H$, and $F$ be given. Assume that there exist $P_1\in\Spnz,P_2\in\Spny$, $\delta>0$, and $\eta>0$ such that the following conditions hold
\begin{align}
\label{eq:Mdbis}
&\mathsf{M}(\mu_i)\prec 0&\forall i\in\{1,2\},
\\
&\begin{bmatrix}
-P_1&P_1-C\tr F\tr P_1\\
\bullet&-e^{\delta T_1}P_1
\end{bmatrix}\preceq 0,
\label{eq:MdDT}
\end{align}
where, for all $\mu\in\R$, $\mathsf{M}(\mu)$ is defined in \eqref{eq:Mmutau} (at the top of the next page)
and $\mu_1\coloneqq\eta$ and $\mu_2\coloneqq(1+\eta)e^{\delta T_2}-1$. Then, functions 
\begin{equation}
\begin{array}{lcl}
     (\varepsilon, \tau)\mapsto V_1(\varepsilon,\tau)&=&e^{-\delta \tau}\varepsilon\tr P_1 \varepsilon,\\
     (\varepsilon, \tilde \theta)\mapsto V_2(\tilde\theta,\tau)& =&(1+\eta -e^{-\delta \tau})\tilde\theta\tr P_2\tilde \theta,
     \end{array}
        \label{eq:V1V2}
\end{equation}
satisfy Assumption~\ref{Assumption:L2} and the set $\mathcal{A}$ in  \eqref{eq:P2:Chap3:A} is GES for  $\mathcal{H}_e$.
\end{theorem}
\begin{proof}
As a first step, notice that $V_1$ and $V_2$ satisfy items (A1) and (A2) in Assumption~\ref{Assumption:L2} with: $\alpha_2=\lambda_{\max}(P_1)$, $\alpha_1=e^{-\delta T_2}\lambda_{\min}(P_1)$,  $\omega_1=\eta\lambda_{\min}(P_2)$, and $\omega_2=(1-e^{-\delta T_2}+\eta)\lambda_{\max}(P_2)$.
Straightforward calculations show that 
\begin{equation}\label{eq:gradient}
\langle\nabla V(x), f(x)\rangle =e^{-\delta\tau}\begin{bmatrix}\ep\\\thetatilde\end{bmatrix}\tr\mathsf{M}(\mu(\tau))\begin{bmatrix}\ep\\ \thetatilde\end{bmatrix},\quad \forall x\in\mathcal{C}
\end{equation}
where, for all $\tau\in[0, T_2]$, $\mu(\tau)\coloneqq(1+\eta)e^{\delta \tau}-1$ and $\mathsf{M}(\cdot)$ is defined in \eqref{eq:Mmutau} (at the top of the next page).
\begin{figure*}
\smallskip

\begin{equation}
\mathsf{M}(\mu)\coloneqq
\begin{bmatrix}
\He(P_1(A\!-\!LC))\!+\!\delta P_1& \!\!P_1L\!+\!\mu(CA\!-\!CLC\!-\!HC)\tr\!P_2\\
\bullet&\mu\He(P_2(CL+H))\!-\!\delta P_2
\end{bmatrix}.
\label{eq:Mmutau}
\end{equation}
\end{figure*}
Since $\mathsf{M}(\mu(\tau))$ is affine with respect to $\mu(\tau)$, it is also convex with respect to it. In addition, notice that $\range\mu=[\eta,\  (1+\eta)e^{\delta T_2}-1]=\colon [\mu_1, \mu_2]$. Therefore, the following equivalence holds:
$$
\mathsf{M}(\mu(\tau))\prec 0,\ \forall \tau \in [0,\ T_2]\  \Leftrightarrow \ 
\mathsf{M}(\mu)\prec 0,\ \mu\in\{\mu_1, \mu_2\}.
$$
Hence, it follows that \eqref{eq:Mdbis} implies item (A4) in 
Assumption~\ref{Assumption:L2}.
To conclude the proof, it remains to show that inequality \eqref{eq:MdDT} implies the satisfaction of item (A3) in 
Assumption~\ref{Assumption:L2}. To this end, notice that for all $\ep\in\R^{n_z}$, $\nu\in[T_1, T_2]$,
$$
\begin{aligned}
&V_1((I-FC)\ep, \nu)\!-\!V_1(\ep, 0)=\\
&\qquad\qquad\ep\tr\left(e^{-\delta\nu} (\Id-FC)\tr P_1 (\Id-FC)-P_1\right)\ep\leq \ep\tr\mathsf{Q}\ep
\end{aligned}
$$
where $\mathsf{Q}\coloneqq e^{-\delta T_1} (\Id-FC)\tr P_1 (\Id-FC)-P_1$. Hence, if
$\mathsf{Q}\preceq 0$, it follows that item (A3) in 
Assumption~\ref{Assumption:L2} holds with any\footnote{It is straightforward to check that $\alpha_2^{-1}\vert\lambda_{\max}(\mathsf{Q})\vert\in [0, 1]$.} $\varpi_d\in\left[0,\frac{\vert\lambda_{\max}(\mathsf{Q})\vert}{\alpha_2}\right]\cap [0, 1)$.
At this stage notice that by simple congruence transformations and by Schur complement, \eqref{eq:MdDT} is equivalent to $\mathsf{Q}\preceq 0$. Hence, \eqref{eq:MdDT} implies that item (A3) in Assumption~\ref{Assumption:L2} holds. The proof is concluded by application of Theorem~\ref{thm:GES}.
\end{proof}
\section{Observer Design}
\label{sec:ObDes}
\subsection{Guaranteed Cost Observer Design}
The objective of this section is to transform the stability condition of Theorem \ref{Theorem1bis} into constructive ones. This means that the observer gains appears now as additional decision variables. In this situation, the conditions are no longer LMI. However, the use of simple manipulations inspired from \cite{ferrante2018TAC}
allows to alleviate this drawback. In addition another aspect of this section is to illustrate how the proposed architecture lends itself to a guaranteed cost design, this is not the case for \eqref{eq:CT_OBS}.

Let $\phi$ be any solution to $\mathcal{H}_e$, consider the following cost functional \cite{Ferrante2018cost}:
$$
\begin{aligned}
\mathcal{J}(\phi)\coloneqq&\int_{\dom_t \phi}\!\!\! q_c(\phi(s, j(s)))ds
\!+\!\sum^{\sup\dom_j\phi}_{j=1}\!\!q_d(\phi(t(j), j-1)),
\end{aligned}
$$
where for all $x=(\varepsilon, \thetatilde, \tau)\in\C$, $q_c(x)\coloneqq\ep\tr Q_F\ep$ and $q_d(x)\coloneqq\ep\tr Q_J\ep$, 
with $Q_F, Q_J\in\Sspnz$. In particular, for any $\xi\in\C$, we consider the following cost associated to $\mathcal{H}_e$: 
$$
\mathcal{J}^\star(\xi)=\sup_{\phi\in\mathcal{S}_{\mathcal{H}_e}(\xi)}\mathcal{J}(\phi).
$$
The following result is established.
\begin{theorem}
\label{Theorem2optdesign}
Suppose that there exist $P_1\in\Spnz,P_2\in\Spny$, $Y\in\mathbb{R}^{n_z\times n_y}, X\in\mathbb{R}^{n_y\times n_y}$, and $Z\in\mathbb{R}^{n_z\times n_y}$, $\delta>0$, and $\eta>0$ such that the following conditions hold:
\begin{align}
\label{eq:McbisDesignOpt}
&\mathsf{R}_i\prec 0, &\forall i\in\{1,2\},\\
\label{eq:MdbisdesignOpt}
&\begin{bmatrix}
-P_1\!+\!Q_J&P_1-C\tr Z\tr\\
\bullet&-e^{\delta T_1}P_1
\end{bmatrix}\preceq 0,
\end{align}
where $\mathsf{R}_i$ is defined in \eqref{eq:Q_i} (at the top of the next page) with
 $\mu_1\coloneqq\eta$, $\mu_2\coloneqq(1+\eta)e^{\delta T_2}-1$, $\tilde\mu_1\coloneqq1$, and $\tilde\mu_2\coloneqq e^{\delta T_2}$. 
 \begin{figure*}
\begin{equation}
\mathsf{R}_i\coloneqq\begin{bmatrix}
\He(P_1A\!-\!YC)\!+\!\delta P_1\!+\!\tilde \mu_iQ_F& \!\!Y\!+\!\mu_i(P_2CA\!-\!XC)\tr\\
\bullet&\mu_i\He(X)\!-\!\delta P_2
\end{bmatrix}.
\label{eq:Q_i}
\end{equation}
\end{figure*}
Let
\begin{equation}\label{gaindesignOpt}
L=P_1^{-1}Y,\quad H=P_2^{-1}X-CP_1^{-1}Y,\quad F=P_1^{-1}Z,
\end{equation}
Then, the following items hold:
\begin{itemize}
\item[$(i)$] $\mathcal{A}$ in \eqref{eq:P2:Chap3:A} is GES for $\mathcal{H}_e$;
\item[$(ii)$] For any initial condition $\xi=(\xi_{\ep}, \xi_{\thetatilde}, \xi_\tau)\in\mathcal{C}$, the following inequality holds
$$
\mathcal{J}^\star(\xi)\leq e^{-\delta\xi_\tau}\xi_{\ep}\tr P_1\xi_{\ep}+(1+\eta-e^{-\delta\xi_\tau})\xi_{\thetatilde}\tr P_2\xi_{\thetatilde}.
$$ 

\end{itemize}
\end{theorem}
\medskip
\begin{proof}
Thanks to the definition of the observer gains in \eqref{gaindesignOpt}, we have
$P_1L=Y,$ $P_2(H+CL)=X$ and $P_1F=Z$. Therefore, due to $Q_F$ and $Q_J$ being positive definite, a few calculations allow to show that \eqref{eq:McbisDesignOpt} and \eqref{eq:MdbisdesignOpt} imply, respectively,  
\eqref{eq:Mdbis} and \eqref{eq:MdDT}. Hence, item $(i)$ follows directly from Theorem~\ref{Theorem1bis}. To conclude, let $V$ be defined as in Assumption~\ref{Assumption:L2} with $V_1$ and $V_2$ as in \eqref{eq:V1V2}. By following analogous steps as in the proof of Theorem~\ref{Theorem1bis}, it can be easily 
shown that the satisfaction of \eqref{eq:McbisDesignOpt} implies for all $x\in\C$,
$\langle\nabla V(x), f(x)\rangle+\ep\tr Q_F\ep\leq 0$.
Similarly, the satisfaction of \eqref{eq:MdbisdesignOpt} can be easily shown to imply for all $x\in\mathcal{D}, g\in G(x)$, $V(g)-V(x)+\ep\tr Q_J \ep\leq 0$.
Thus, since from item $(i)$ maximal solutions to $\mathcal{H}_e$ converge to the set $\mathcal{A}$ in \eqref{eq:P2:Chap3:A}, $V$ is positive definite with respect to $\mathcal{A}$ and continuously differentiable on $\R^{n_x}$, direct application of \cite[Corollary 1]{Ferrante2018cost} yields $(ii)$. Hence, the result is established.  
\end{proof}
\subsection{Optimal Design and Numerical Issues}
As mentioned in the introduction, one of the main objectives of the proposed observer consists of reducing the variation of the plant state estimate across jumps. To achieve this goal, it appears relevant to consider additional constraints throughout the design of the observer gains, and more in particular on the gain $F$. The result stated next provides a possible approach towards this goal.
\begin{proposition}
\label{Theorem2optdesign2}
Consider $P_1\in\Spnz,P_2\in\Spny$, $Y\in\mathbb{R}^{n_z\times n_y}, X\in\mathbb{R}^{n_y\times n_y}$, and $Z\in\mathbb{R}^{n_z\times n_y}$, and positive real numbers $\gamma_1$ and $\gamma_2, \delta$, and $\eta$, such that \eqref{eq:McbisDesignOpt},  \eqref{eq:MdbisdesignOpt} and
\begin{equation}
\label{eq:constrain}
\begin{bmatrix}
P_1&\!\!\!\! Y\\
\bullet &\!\!\! \gamma_1 I_{n_y}
\end{bmatrix}\! \succ 0,\quad \begin{bmatrix}
P_1& Z\\
\bullet &\!\!\! \gamma_2 I_{n_y}
\end{bmatrix}\! \succ 0,
\end{equation}
hold. Then, under the selection of the observer gains given in \eqref{gaindesignOpt} the following items hold:
\begin{itemize}
\item[$(i)$] the set $\mathcal{A}$ defined in \eqref{eq:P2:Chap3:A} is GES for $\mathcal{H}_e$;
\item[$(ii)$] For any initial condition $\xi=(\xi_{\ep}, 0 
, \xi_\tau)\in\mathcal{C}$, inequality
$
\mathcal{J}^\star(\xi)\leq \xi_{\ep}\tr P_1\xi_{\ep}$ holds.
\item [$(iii)$] The norm of the observer gains $F$ and $L$ is constrained.
\end{itemize}
\end{proposition}
\begin{proof}
Items $(i)$ and $(ii)$ follow from item $(ii)$ of Theorem \ref{Theorem2optdesign}, whenever $\xi_{\thetatilde}=0$.  Let now $L$, $F$ be selected as in \eqref{gaindesignOpt}. Item $(iii)$ follows from the application of the Schur complement, revealing that both matrix inequalities in \eqref{eq:constrain} are equivalent to $L^{\!\top} \! P_1 L,\!\preceq \! \gamma_1 I_{n_y}$, $F^{\!\top} \!P_1 F\!\preceq\!\gamma_2 I_{n_y}$.
\end{proof}
Proposition~\ref{Theorem2optdesign2} can be embedded into the following optimization problem to perform an optimal design of the observer:
\begin{equation}
\label{eq:ConstrainedOpt}
\begin{aligned}
&\!\underset{P_1,P_2,X,Y,Z,\gamma}{\minimize}&\quad& \Tr(P_1)\!+\!\alpha_1\gamma_1\!+\!\alpha_2\gamma_2\\
&\text{subject to} &      & \eqref{eq:McbisDesignOpt}, \eqref{eq:MdbisdesignOpt},\eqref{eq:constrain}.
\end{aligned}
\end{equation}
In particular, minimizing $\Tr(P_1)+\alpha_1\gamma_1+\alpha_2\gamma_2$ allows to simultaneously bound the observer gains $F$ and $L$ and, in the light of item $(ii)$ in Proposition~\ref{Theorem2optdesign2}, to minimize the cost $\mathcal{J}^\star(\xi)$ (with $\xi_{\tilde \theta}=0$) uniformly with respect to $\xi$. The parameters $\alpha_1$ and $\alpha_2$ are introduced to enable a tradeoff between the constraints on the observers gains $F$ and $L$. Those parameters need to be tuned a priori. The impact of this tuning is discussed in Section~\ref{sec:example}.
\section{Numerical example}
\label{sec:example}
The objective of this section is to showcase the effectiveness of the proposed hybrid observer\footnote{Simulations of hybrid systems are performed in Matlab$^{\tiny{\textregistered}}$ via the \textit{Hybrid Equations (HyEQ) Toolbox} \cite{sanfelice2013toolbox}.}. Consider the following data for \eqref{eq:P2:Chap3:Plant}:
$$
\begin{aligned}
&A=\begin{bmatrix}
0.2 &-1.01\\
 1& 0\\ 
 \end{bmatrix}, C\tr=\begin{bmatrix}0.5 \\-1 \end{bmatrix}, T_1=0.5, T_2=1.1.
\end{aligned} 
$$
TABLE~\ref{tab:gains} gathers the gains obtained by solving optimization problem \eqref{Theorem2optdesign2} with $\delta=0.03$ and $\eta=10^{-4}$ and for several values of $\alpha_1$ and $\alpha_2$. Noticing that the case $\alpha_1=\alpha_2=0$ refers to the situation in which no constraints on the observer gains are imposed, which leads to overly large gains\footnote{\textcolor{blue}{In this example, $Q_F=I$ and $Q_J=0.01 I$. In the published paper, these values have been inadvertently not indicated. }}.

Moreover, to show the benefit of the proposed observer in ensuring convergence speed while limiting the variation of the estimate across jumps, we compare it with the observer \eqref{eq:CT_OBS}. Specifically, we consider a ``large'' (4th row in TABLE~\ref{tab:gains}) and a ``small'' gain (5th row in TABLE~\ref{tab:gains}), both gains are designed via the conditions in \cite{Ferrante2016state}. Indeed, for \eqref{eq:CT_OBS}, the amplitude of the variations of the estimate can be limited by minimizing the norm of the gain $F$. Notice that since observer  \eqref{eq:CT_OBS} runs in open-loop in between measurements, for such a scheme it is not possible to perform a guaranteed cost design as done for observer proposed in this paper via Theorem~\ref{Theorem2optdesign}. This explains why the last column of TABLE~\ref{tab:gains} contains $\emptyset$. 
\begin{table}[t]
\begin{center}
\begin{tabular}{|c|c|c|c|c|}
\hline
Cases & $L$ & $F$ & $H$ &$\Tr(P_1)$\\[2.5mm]
\hline
\!\!\!\!\!$\begin{array}{l}
\mbox{(I): \eqref{eq:ConstrainedOpt}}\\
\alpha_1\!=\!0\\
\alpha_2\!=\!0.
\end{array}$\!\!\!\!\!\!
&\!\!\!\! $\left[\begin{smallmatrix}
10877 \\ 
-98807\\ 
\end{smallmatrix}\right]$\! & \!\!\!
$\left[\begin{smallmatrix}
0.104\\ 
-0.948\\ 
\end{smallmatrix}\right]$ \!\!\!& \!\!\!$-104250 $\!\!\!&\!\!\! $353.8$\\
[2.5mm]
\hline
\!\!\!\!\!$\begin{array}{l}
\mbox{(II): \eqref{eq:constrain}}\\
\alpha_1\!=\!100\\
\alpha_2\!=\!0.1.
\end{array}$\!\!
&\!\!\!\! $\left[\begin{smallmatrix}
3.68\\  -24.47\\ 
\end{smallmatrix}\right]$ & \!\!\!
$\left[\begin{smallmatrix}
0.104\\ 
-0.948\\ 
\end{smallmatrix}\right]$\!\!\! & \!\!\! $-25.93 $& \!\!\! $354.7$\\
[2.5mm]
\hline
\!\!\!\!\!$\begin{array}{l}
\mbox{(III): \eqref{eq:constrain}}\\
\alpha_1\!=\!100\\
\alpha_2\!=\!1.
\end{array}$\!\!
&\!\!\!\! $\left[\begin{smallmatrix}
3.68\\  -24.47\\ 
\end{smallmatrix}\right]$ & \!\!\!
$\left[\begin{smallmatrix}
0.040\\ 
-0.364\\ 
\end{smallmatrix}\right]$\!\!\! & \!\!\! $-11.47 $& \!\!\! $357.1$\\
[2.5mm]
\hline
\!\!\!\!\!$\begin{array}{l}
\mbox{(IV): Hybrid obs. \cite{Ferrante2016state}}\\
L=0,\ H=0
\end{array}$\!\!
&\!\!\!\! $\left[\begin{smallmatrix}
0\\  0\\ 
\end{smallmatrix}\right]$ & \!\!\!
$\left[\begin{smallmatrix}
0.097\\
-0.905\\ 
\end{smallmatrix}\right]$\!\!\! & \!\!\! $0 $& \!\!\! $\emptyset$\\
[2.5mm]
\hline
\!\!\!\!\!$\begin{array}{l}
\mbox{(V): Hybrid obs. \cite{Ferrante2016state}}\\
L=0,\ H=0\\
\end{array}$\!\!
&\!\!\!\! $\left[\begin{smallmatrix}
0\\  0\\ 
\end{smallmatrix}\right]$ & \!\!\!
$\left[\begin{smallmatrix}
0.183\\
-0.333\\
\end{smallmatrix}\right]$\!\!\! & \!\!\! $0 $& \!\!\! $\emptyset$\\
[2.5mm]
\hline
\end{tabular}
\end{center}
\label{default}
\caption{Different selections of the observer gains.}
\label{tab:gains}
\end{table}%
To show the effectiveness of the proposed design, in Fig.~\ref{fig:eps}  we compare the evolution of the estimation error $\varepsilon$ with the observer gains presented in Table~\ref{tab:gains} (Cases III, IV and V) from the initial condition $z(0, 0)=\begin{bsmallmatrix}10\\ 0\end{bsmallmatrix}$, $\hat{z}(0,0)\!=\!\begin{bsmallmatrix}0\\0\end{bsmallmatrix}$, $\theta(0,0)\!=\!Cz(0,0)\!=\!5$ and
$\tau(0,0)\!=\!0$. In these simulations, the value of $\tau$ at jumps is selected as $\tau(t, j+1)=\frac{T_2-T_1}{2}\sin(t)+\frac{T_2+T_1}{2}$.

Fig.~\ref{fig:eps} clearly shows that for the observer \eqref{eq:CT_OBS} limiting the norm of the gain $F$ induces poor convergence time. On the other hand, the proposed observer enables to limit the variation of the estimation error across jumps (this is mostly visible in the evolution of $\ep_2$) while maintaining a fast convergence rate. 
\begin{figure}
\psfrag{t}[1][1][1]{$t$ [sec]}
\psfrag{y}[1][1][1]{$\varepsilon_1$}
\psfrag{y2}[1][1][1]{$\varepsilon_2$}
\includegraphics[trim=1cm 0cm 0cm 0.1cm, clip, width=1\columnwidth]{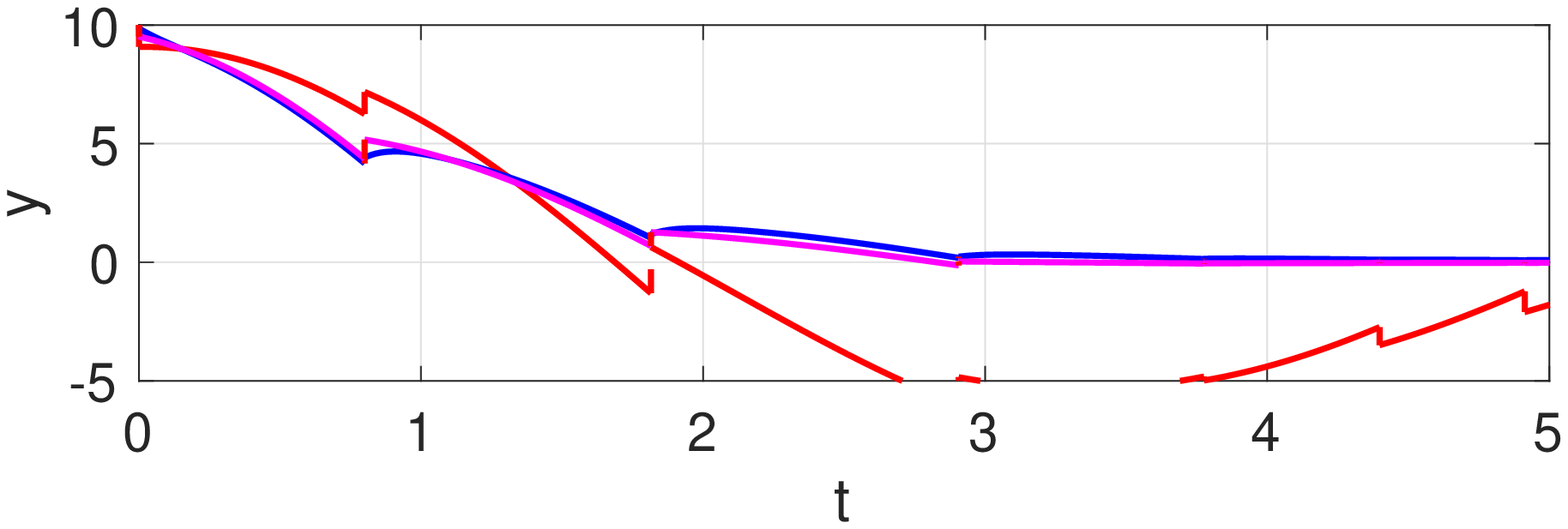}
\includegraphics[trim=1cm 0cm 0cm 0.1cm, clip, width=1\columnwidth]{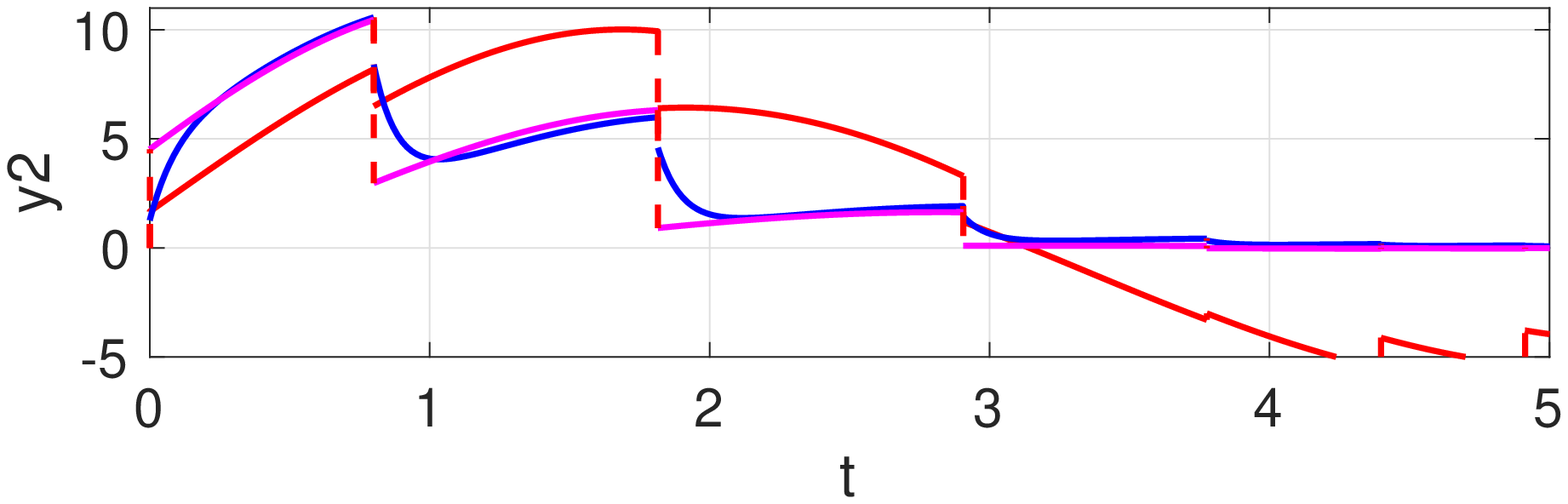}
\caption{Evolution of the estimation error ($\ep_1$ and $\ep_2$) (projected onto ordinary time) for the observers provided in TABLE~\ref{tab:gains}. The blue, purple, and red lines refer to cases (III) to (V), respectively, with the same order.}
\label{fig:eps} 
\end{figure}
\section{Conclusion} 
In this paper, a novel observer design for linear systems subject to aperiodic sampled-data measurements has been presented. The estimation error dynamics are modeled as a hybrid dynamical systems. By employing a Lyapunov approach, sufficient conditions for global exponential stability of a closed set wherein the estimation error is zero are obtained. Guaranteed cost optimal design of the observer gains is presented as the solution of an LMI optimization problem. The potential of this new hybrid observer is illustrated through an academic example.

This paper can be seen as a first step towards the derivation of more general observers for systems subject to aperiodic sampled-data measurements. One of the main features of the proposed architecture consists of combining two types of injections. The use of this additional degree of freedom provides more flexibility in the design of the observer and may potentially lead to better tradeoff between robustness to measurement noise and convergence speed. In addition, with the objective of limiting the variation of the state across jumps, we envision to explore the use of explicit hard bounds on the injection term of the observer.  Another direction consists in relaxing the assumption on the aperiodic samplings by proposing an average dwell-time assumption, which would lead to less conservative conditions. Finally, an interesting direction pertains to the use of less conservative clock-dependent Lyapunov functions for the analysis of the estimation error dynamics. In this setting, an adaptation of the results in \cite{briat2013convex} to the observer design problem seems promising.
\balance
\bibliographystyle{plain}
\bibliography{biblio}
\end{document}